\documentclass[aps,prl,twocolumn,superscriptaddress,floatfix,letter]{revtex4}

\usepackage{amsmath,amssymb,amsthm}
\usepackage{times}
\usepackage{graphicx}

\newtheorem{lem}{Lemma}
\newtheorem*{thm}{Theorem}

\newenvironment{addendum}{%
    \setlength{\parindent}{0in}%
    \small%
    \begin{list}{Acknowledgements}{%
        \setlength{\leftmargin}{0in}%
        \setlength{\listparindent}{0in}%
        \setlength{\labelsep}{0em}%
        \setlength{\labelwidth}{0in}%
        \setlength{\itemsep}{12pt}%
        }
    }
    {\end{list}\normalsize}

\newcommand{\rmsubscr}[1]{\textrm{#1}}

\newcommand{\eqendspace}{\;}
\DeclareMathOperator{\Var}{Var}
\DeclareMathOperator{\Mean}{Mean}
\DeclareMathOperator{\Span}{Span}
\DeclareMathOperator{\Trace}{Tr}
\DeclareMathOperator{\Prediction}{Prediction}
\DeclareMathOperator{\MSE}{MSE}

\newcommand{\MC}{\textrm{\upshape MC}}
\newcommand{\GGE}{\textrm{\upshape GGE}}

\newcommand{\lb}{\langle}
\newcommand{\rb}{\rangle}

\usepackage[dvips]{color}

\makeatletter

\begin{document}

\title{Geometry of quantum observables and
thermodynamics of small systems}

\author{Maxim Olshanii}
\affiliation{Department of Physics, University of Massachusetts Boston, Boston, MA 02125, USA}

\begin{abstract}
The concept of ergodicity---the convergence of the temporal
averages of observables to their ensemble averages---is the cornerstone
of thermodynamics. The transition from a predictable, integrable
behavior to ergodicity is one of the most difficult
physical phenomena to treat; the celebrated KAM theorem is the
prime example. This Letter is founded on the observation that
for many classical and quantum observables,
%
\begin{center}
\begin{minipage}{4.7in}
%
the sum of the ensemble
{\it variance} of the temporal average and the ensemble average of
temporal {\it variance} remains constant
across the integrability-ergodicity transition.
%
\end{minipage}
\end{center}
%
We show that this property induces
a particular geometry of quantum observables---Frobenius (also known as Hilbert-Schmidt) one---that
naturally encodes all the phenomena associated with the emergence of ergodicity:
the Eigenstate Thermalization effect \cite{deutsch1991,srednicki1994,rigol2008},
the decrease in the inverse participation ratio \cite{santos2009},
and the disappearance of the integrals of motion. As an
application, we use this geometry to solve a known problem
of optimization of the set of conserved
quantities---regardless of whether it comes from symmetries or from finite-size effects---to be incorporated in
an extended thermodynamical theory of integrable, near-integrable, or mesoscopic systems
\cite{rigol2007,rigol2011_140405,kollar2011_054304,gring2012_Science,gramsch2012_PRA}.
\end{abstract}

\maketitle

Consider a ball in a classical rectangular billiards with periodic boundary conditions, first without and then with a strong localized obstacle inside. (When present, the obstacle makes the system ergodic.) In both cases, consider two kinds of problems: 1. given a specific initial velocity of the ball, what are the subsequent temporal fluctuations in (say) the $x$-component of the velocity? 2. draw a bunch of initial velocities from a thermal distribution;  for each, compute the infinite time average of the $x$-component of the velocity; what is the variance among these infinite time averages?

If the obstacle is absent, both components of the velocity vector will
remain equal to their initial values forever, and so the temporal fluctuations are zero; as far as the second question, the variance of the infinite time averages is equal to that in the thermal distribution, and thus is large. On the other hand, if the obstacle is present, the situation is reversed (this will be so provided the ball keeps hitting the obstacle, which will happen provided the ratio of the velocity components, $v_{y}/v_{x}$, is incommensurate with the ratio of the billiard lengths, $L_{y}/L_{x}$; and this will hold with probability one as long as the initial velocities as sampled from some continuous probability distribution): now the temporal fluctuations in the velocity component are large no matter what the initial velocity; at the same time, the infinite time average of the velocity component is the same no matter what the initial velocity (namely, it is zero), and so the variance of the infinite time averages is zero. Moreover, as we will show, the shot-to-shot variance in the (exactly integrable) case of an empty billiards equals the temporal variance in the ergodic case.

Given some imagination, one may suspect a ``conservation law'' acting across the transition from
an integrable to an ergodic system as one increases the strength of the integrability-breaking perturbation.
Indeed, Fig.~\ref{Olshanii_fig1}a shows that in the case of square billiard
perturbed by a soft localized potential barrier in the middle \cite{supmat}, the sum of the microcanonical
variance of the time average, $\Var_{\MC}[\Mean_{t}[A]]$,
and the shot-to-shot average of the temporal variance, $\Mean_{\MC}[\Var_{t}[A]]$,
with the observable $A$ being the difference between the horizontal and vertical kinetic energies, remains
the same for the heights of the barrier less than or comparable to the kinetic energy.
For all points,
the statistical ensembles used were microcanonical ensembles
with the same phase-space volume ($W=1184.3$) covered and with the same phase-space volume
($W_{b}=7895.7$)
occupied by the phase-space points with energies below the lower-energy boundary of the
microcanonical window. Such a set of microcanonical ensembles is the closest classical analogue of
a quantum set that uses the same window of quantum state {\it indices} for all perturbation strengths used.


Let us now try to translate the conjecture expressed in the first paragraph of this Letter to quantum language.

Naively, one may try to replace the ensemble of classical initial conditions by an ensemble of random superpositions
of quantum eigenstates. However, consider the limiting case
of an integrable system. There, a single
given initial state will already cover a variety of---generally unrelated---sets of integrals of motion. To the
contrary, in the classical case, a given initial point corresponds to a single set. Accordingly, we suggest
an ensemble of randomly chosen eigenstates of the Hamiltonian as the ensemble of initial conditions. In particular,
the ensemble variance of the temporal means will be translated to the quantum language as
%
\begin{multline}
%
\Var_{\MC}[\Mean_{t}[A]]\Big{|}_{\textrm{QM}}
\equiv
\Var_{\MC}[\lb\alpha\vert \hat{A} \vert\alpha\rb]
\\
\qquad =
(N_{\MC})^{-1} \sum_{\alpha \in [\MC]} (\lb\alpha\vert \hat{A} \vert\alpha\rb - \lb A \rb)^2
\eqendspace,
\label{Var_MC_Mean_t}
\end{multline}
%
where $\lb A \rb = \Mean_{\MC}[A] \equiv (N_{\MC})^{-1} \sum_{\alpha \in \MC} \lb\alpha\vert \hat{A} \vert\alpha\rb $
is the ensemble mean of the observable $\hat{A}$,
$\sum_{\alpha \in [\MC]} \ldots \equiv \sum_{\alpha \in [\alpha_{min},\,\alpha_{max}]} \ldots$ is a sum over a
microcanonical window bounded by some energies $E_{min}$ and $E_{max}$, and $E_{\alpha}$ is the energy spectrum
of the system. The vanishing of the fluctuations (\ref{Var_MC_Mean_t}) in the thermodynamic limit---the so-called
Eigenstate Thermalization effect---is sufficient
\cite{deutsch1991,srednicki1994,rigol2008} for the emergence of ergodicity. Our Letter, in part, aims to
devise a scale for these fluctuations that determines if a given observable is closer, by behavior, to an
integral of motion or to a thermalizable observable.

The question of a proper quantum analogue of the temporal fluctuations is both more involved and better studied.
A complication arises from the fact that a quantum state is altered after each measurement. However, consider
the following procedure: for every realization of the initial state of interest, the observable is measured only
once, and then the initial state is prepared again. For every instant of time of interest, the observable is
measured several times, and then, another instant of time is addressed. It has been argued\cite{srednicki1996_L75} that this
\begin{widetext}
\mbox{}
\begin{figure}[h]
\begin{center}
\includegraphics[scale=.8]{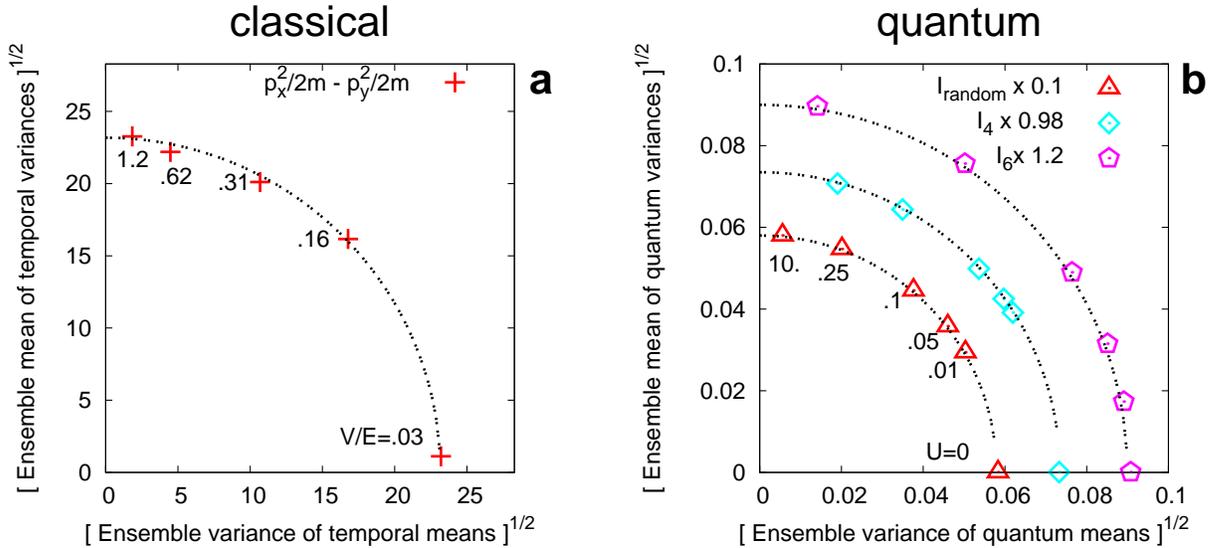}
\end{center}
\vspace{-0.7cm}
\caption{  \label{Olshanii_fig1}
%
%
{\bf Two variances}. The units are such that $m = L/(2\pi) =1$, where $m$ is the particle mass, and $L$ is the side of the billiard.
{\bf a},
Ensemble mean of the temporal variance vs.\ the ensemble variance of the temporal
mean of the difference between the $x-$ and $y-$kinetic energies for
rectangular billiard with periodic boundary conditions perturbed by a soft-core barrier.
At zero barrier hight, the mean energy over the ensemble is $E_{0}=33.7$ .
The points plotted are labeled by the ratio, $V/E$, of the barrier hight to the mean energy of the ensemble
(that includes the energy of the barrier).
{\bf b},
Ensemble mean of the quantum variance vs.\ the ensemble variance of the
quantum expectation values for three integrals of motion of a system of hard core bosons, where the integrals of motion are partially destroyed by adding a soft-core two-body repulsive potential.
For this figure: the number of particles is $N=4$, the number of lattice sites is $L=16$, and we impose open
boundary conditions. The soft-core interaction potential has a constant hight $U$ at distances of
four sites or less, and it is zero otherwise.
$
\hat{I}_{4}
$
and
$
\hat{I}_{6}
$
(see the text for expressions)
are the integrals of motion related to the forth  and sixth moments of the momentum distribution of the underlying free fermions.
$\hat{I}_{\textrm{random}}$ is an artificial integral of motion represented by a diagonal---in the basis of the unperturbed eigenstates---matrix with random entries, uniformly distributed between $-1$ and $+1$. For this observable, the square cosine of the angle between a point on the figure and the horizontal axis equals (up to small corrections of the order of $(N_{\MC})^{-1}$) the inverse participation ratio $\eta $ (see Eq.~(\ref{eta})). Observe also that the behavior of the two other observables is qualitatively simlar.
}
\end{figure}
\end{widetext}
procedure is indeed the most suitable quantum counterpart of the classical temporal fluctuations along
a trajectory. Furthermore, it has been suggested that in case of an ergodic motion, these fluctuations are
nothing else but the {\it thermal} fluctuations in the system\cite{srednicki1996_L75}. According to this scenario,
the quantum uncertainty in the results of the measurements is the way a quantum system emulates the classical instability
with respect to the initial conditions: in both cases, the outcome of a single measurement is irreproducible,
fundamentally so in the quantum case, and operationally so in the classical one. Accordingly, we define the quantum analogue
of the ensemble mean of the temporal variance as
%
\begin{multline}
%
%
\Mean_{\MC}[\Var_{t}[A]]\Big{|}_{\textrm{QM}}
\\
 \equiv
%
\Mean_{\MC}[\lb\alpha\vert \hat{A}^2 \vert\alpha\rb - \lb\alpha\vert \hat{A} \vert\alpha\rb^2]
\\
=
(N_{\MC})^{-1} \sum_{\alpha \in [\MC]} \lb\alpha\vert \hat{A}^2 \vert\alpha\rb - \lb\alpha\vert \hat{A} \vert\alpha\rb^2
\eqendspace.
\label{Mean_MC_Var_t}
\end{multline}
%

Consider now an integrable system, with hamiltonian $\hat{H}_{0}$
perturbed by a non-integrable perturbation $\hat{V}$. The full hamiltonian reads
$
\hat{H} = \hat{H}_{0} + g \hat{V}
$,
where the parameter $g$ determines the degree by which the integrability is broken. Let the states
$\vert \alpha\rb $ be the eigenstates of the full hamiltonian.
The conjecture in the first paragraph
of the Letter then reads:
%
\begin{multline}
\Var_{\MC}[\Mean_{t}[A]]\Big{|}_{\textrm{QM}}
+
\Mean_{\MC}[\Var_{t}[A]]\Big{|}_{\textrm{QM}}
\\
= \text{a constant independent of $g$}
\eqendspace.
\label{conjecture}
\end{multline}
%
%

The data presented in Fig.\ \ref{Olshanii_fig1}b test the conjecture (\ref{conjecture}) using the example
of a one-dimensional gas of lattice hard-core bosons perturbed by an added two-body soft-core
repulsive interaction. In one dimension, both continuous-space and lattice hard-core bosons are known
to be integrable: in both cases,
there is a map (Girardeau's map\cite{girardeau1960_516}, and the Jordan-Wigner transformation, respectively) that connects
the eigenstates of the system to the eigenstates of a free Fermi gas. The integrals of motion are thus represented
by the occupation numbers of the eigenstates of the one body hamiltonian for the particles of the underlying free Fermi gas.

We analyze the decay of the fourth and sixth integrals of motion,
\[
\hat{I}_{4}
= \frac{1}{2L} \sum_{j=1}^{L-2} (
                             (\hat{a}^{\dagger}_{j}\hat{a}^{}_{j+2} + h.c.)
                             -
                             (\hat{a}^{\dagger}_{1}\hat{a}^{}_{1} + \hat{a}^{\dagger}_{L}\hat{a}^{}_{L})
                          )
\,,
\]
\begin{multline*}
\hat{I}_{6}
= \frac{1}{2L} \sum_{j=1}^{L-3} (
                                 (\hat{a}^{\dagger}_{j}\hat{a}^{}_{j+3} + h.c.)\\
                                 -
                                 ((\hat{a}^{\dagger}_{1}\hat{a}^{}_{2}+h.c. + (\hat{a}^{\dagger}_{L-1}\hat{a}^{}_{L}+h.c.))
                               )
\,,
\end{multline*}
where
$\hat{a}^{}_{j}$ is the $j$-th site annihilation free-fermionic operator
(see also the Supplementary Discussion), as we increase the strength of the non-integrable
perturbation. One can see that while the quantum (the analogue of thermal) fluctuations gradually increase,
the deviations from the ergodicity decrease. However,
in accordance with the conjecture (\ref{conjecture}), the sum of the two variances remains approximately constant.
The square of the radius of the circles corresponds to the ensemble variance of the observable, over a series of {\it single}
measurements---with no subsequent quantum or temporal averaging---on a randomly chosen eigenstate.

Let us now reveal the intuition behind the conjecture (\ref{conjecture}). The left-hand-side of the
relationship (\ref{conjecture}) can be written, by rearranging the terms in Eqs.~\ref{Var_MC_Mean_t} and \ref{Mean_MC_Var_t},
as the {\it ensemble variance} of the observable $\hat{A}$:
%
\begin{multline}
%
\Var_{\MC}[\Mean_{t}[A]]\Big{|}_{\textrm{QM}} + \Mean_{\MC}[\Var_{t}[A]]\Big{|}_{\textrm{QM}}
\\
= \Mean_{\MC}[A^2]-\Mean_{\MC}[A]^2
\eqendspace.
\label{ensemble_variance}
\end{multline}
%

In turn, the ensemble variance on the right-hand-side of the above equation is a function of two ensemble means; and the key realization is that these ensemble means remain constant as the coupling constant $g$ of the integrability-breaking perturbation is increased from zero up to a value where the system first becomes ergodic. Indeed remark that typically, ergodicity emerges for an interaction strength much weaker than the one required to alter the thermal expectations of observables. For example, even though the van der Waals interactions between the molecules constituting air do not modify the Maxwell distribution, the interactions are strong enough to lead to a Maxwell distribution from any initial state.

Mathematically, the general criterion for the independence of the ensemble means on $g$ may be stated as follows. For every $g$, there is a characteristic energy interval $\delta E(g)$ such that the integrability-breaking perturbation appreciably couples only those eigenstates of the integrable hamiltonian whose energy difference is less than $\delta E(g)$. Our criterion is that $\delta E(g)$ be much smaller than the energy width of the microcanonical window used to define the ensemble means. The reason is as follows: consider the eigenstates of the unperturbed system. The perturbation results in their mutual coupling; but if the criterion is fulfilled, then we may neglect the coupling of states within the window to states outside the window. But if we do that, then we may, in fact, truncate our Hilbert space to just the states inside the window (let us call the corresponding hamiltonian a `truncated hamiltonian'). Indeed, let us first truncate, and then turn on the perturbation. In that case, the eigenstates of the perturbed truncated hamiltonian are related to those of the unperturbed truncated hamiltonian by a unitary transformation. And now note that the two ensemble means are in fact traces, and thus do not change under unitary transformation of the basis; thus they do not depend on $g$.



Observe now that in the right hand side of the relationship in Eq.~(\ref{ensemble_variance}), we find a quadratic polynomial, built out of the matrix elements of the observable $\hat{A}$, that is approximately invariant with respect to changes in the parameter $g$; in contrast, on the left-hand side is a sum of two quadratic polynomials that both vary with $g$. This realization inspires one to seek a geometric meaning of the relationship in Eq.~(\ref{ensemble_variance}), where
the right hand side is the square norm of an unknown vector that is linearly related to the operator $\hat{A}$, and the left hand side is a decomposition
of this square norm over some complementary subspaces (that change with the perturbation strength $g$) of the linear space the vector belongs to.
From what follows, we will see that this is almost what happens, with the exception of the $\Mean_{\MC}[A]^2$ term that belongs rather to the
left-hand side.
\begin{figure}
\begin{center}
\includegraphics[scale=.3]{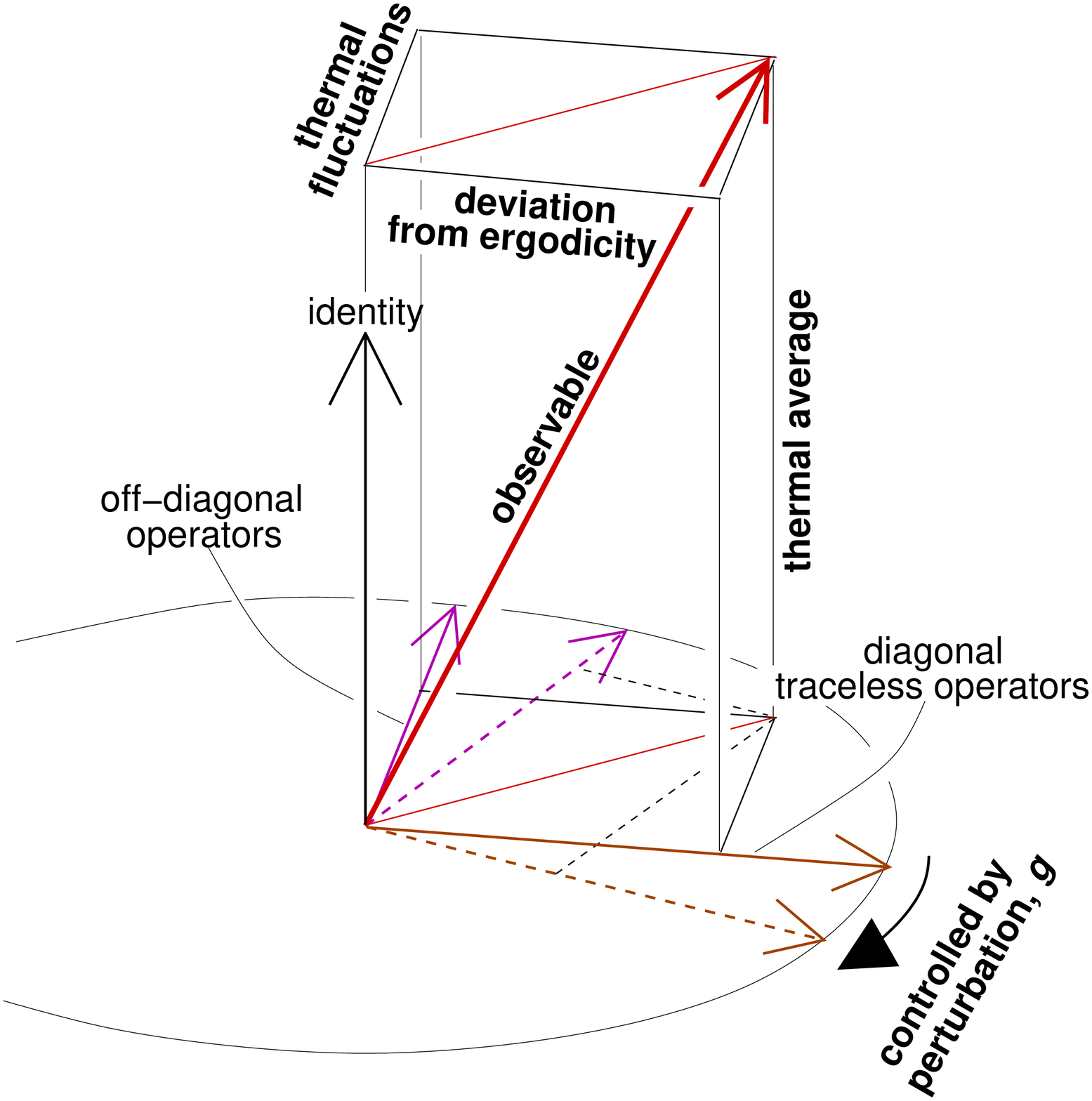}
\end{center}
\vspace{-0.7cm}
\caption{\label{Olshanii_fig2}
%
%
{\bf The Hilbert-Schmidt geometry and quantum integrability-ergodicity transition}.
The $N^2$-dimensional
space of quantum operators acting in an $N$-dimensional
Hilbert space is divided onto a sum of three subspaces: a one-dimensional space spanned by the
identity operator, an $(N-1)$-dimensional space, ${\cal L}_{\rmsubscr{tl,\,d-}\hat{H}}$, of traceless
diagonal---in the basis of the eigenstates of the hamiltonian in question---real matrices, and
an $N(N-1)$-dimensional space of the purely off-diagonal---in the above sense---Hermitian matrices, ${\cal L}_{\rmsubscr{o-d-}\hat{H}}$.
If the coupling constant $g$ in the system hamiltonian
$
\hat{H} = \hat{H}_{0} + g \hat{V}
$
changes, the ${\cal L}_{\rmsubscr{tl,\,d-}\hat{H}}$ and ${\cal L}_{\rmsubscr{o-d-}\hat{H}}$ undergo a rotation, while the identity axis, the projection to which is proportional to the
microcanonical average, remains the same. However, the Hilbert-Schmidt vector corresponding to a given
observable of interest remains fixed along with
the identity axis (and thus also along with the
microcanonical average). Observables with a large projection onto ${\cal L}_{\rmsubscr{tl,\,d-}\hat{H}}$ correspond to quasi-conserved,
non-ergodic quantities. In the opposite case, a large projection onto the ${\cal L}_{\rmsubscr{o-d-}\hat{H}}$ space signifies a thermalizable observable,
whose infinite time average coincides with its thermal expectation value. Ergodicity is reached when ${\cal L}_{\rmsubscr{o-d-}\hat{H}}$ aligns
with the traceless versions of all empirically relevant observables---e.g. with all
one-, two-, and three-body obseravbles in a many-body system.
}
\end{figure}
The inner product that generates the anticipated geometric structure is the Frobenius, or Hilbert-Schmidt, inner product.

The Frobenius or Hilbert-Schmidt (HS) inner product between two matrices reads:
\begin{eqnarray}
(\hat{A} | \hat{B}) \equiv \Trace[\hat{A}^{\dagger} \hat{B}]
\eqendspace.
\label{HS}
\end{eqnarray}
Quantum observables, represented by Hermitian matrices, form a linear space over the field of real numbers; the product in Eq.~(\ref{HS})
induces a real-valued inner product on this space. Observe also that this product is invariant under unitary transformations $\hat{A}\mapsto\hat{U} \hat{A} \hat{U}^{-1}$, where $\hat{U}$ is a unitary matrix. Thus the unitary transformations form a subgroup of all possible linear transformations that preserve the product.

The Hilbert space where we are going to deploy the HS structure is a microcanonical window of eigenstates, of size $N_{\MC}$.
Let us introduce an $N_{\MC}^2$-dimensional linear space, over the field of real numbers, of all Hermitian operators (generally represented by complex matrices) acting within
an $N_{\MC}$-dimensional microcanonical window. The HS product in Eq.~(\ref{HS}) constitutes an inner product on this space.
This space can be conveniently decomposed into a direct sum 
of three pairwise orthogonal
subspaces. The first one is a one-dimensional space, ${\cal L}_{U}$, spanned by the identity operator. The second, ${\cal L}_{\rmsubscr{tl,\,d-}\hat{H}}$, is the space of all
traceless {\it integrals of motion}; i.e.\ the space of traceless operators that are purely
diagonal in the basis of the eigenstates $\vert \alpha \rb$.
Its dimension is $N_{\MC}-1$. The third one, ${\cal L}_{\rmsubscr{o-d-}\hat{H}}$, of dimension $N_{\MC}(N_{\MC}-1)$, is the space of all purely off-diagonal operators
in the same basis. Now, Eq.~(\ref{ensemble_variance}) can be
rewritten as
$
\Mean_{\MC}[A^2] = \Mean_{\MC}[A]^2 + \Var_{\MC}[\Mean_{t}[A]]\Big{|}_{\textrm{QM}} + \Mean_{\MC}[\Var_{t}[A]]\Big{|}_{\textrm{QM}}
$
and interpreted as a relationship between the norm of the operator $\hat{A}$ and its projections to the members of a complete set of orthogonal
subspaces.

The HS angle between the traceless version of a given observable, $\hat{A}_{\rmsubscr{tl}} \equiv \hat{A} - \Trace[\hat{A}]/N_{\MC}$,
and the space of the traceless integrals of motion constitutes a useful
measure that defines where the observable lies on the ``integral of motion'' vs.\
``thermalizable observable'' scale. Indeed, the tangent square of this angle is just
the ratio between (the quantum analogues of) the temporal fluctuations and the time-average vs.
ensemble average discrepancy:
\begin{eqnarray}
\tan^2[\theta_{{\cal L}_{\rmsubscr{tl,\,d-}\hat{H}} ,\, \hat{A}_{\rmsubscr{tl}}}] =
\frac{\Mean_{\MC}[\Var_{t}[A]]\Big{|}_{\textrm{ QM}}}{\Var_{\MC}[\Mean_{t}[A]]\Big{|}_{\textrm{QM}}}
\label{integrability_angle}
\end{eqnarray}
(see Fig.\ \ref{Olshanii_fig2} for an illustration).
Here and below, the angle between an HS vector $\hat{B}$ and an HS hyperplane ${\cal L}$
is defined through $\cos^2(\theta_{\hat{B},\,{\cal L}}) = \sum_{i} \cos^2(\theta_{\hat{B},\,\hat{e}_{i}});\,0\le \theta_{\hat{B},\,{\cal L}} \le \pi/2$,
where $\{\hat{e}_{i}\}$ is any orthonormalized basis set in ${\cal L}$.
If Eigenstate Thermalization \cite{deutsch1991,srednicki1994,rigol2008}
holds, the angle (\ref{integrability_angle}) approaches $90^{\circ}$.

The inverse participation ratio $\eta $, a measure closely related to the transition to thermal behavior
\cite{georgeot1997,santos2009,olshanii2012_641}, also acquires a clear geometric meaning. It defines
the average angle between the projector to an eigenstate of the integrable system $\hat{H}_{0}$ and the space of the
integrals of motion of the new system $\hat{H}$:
%
\begin{multline}
\mbox{}\qquad
\eta
\equiv
N_{\MC}^{-1} \sum_{\alpha,\,\alpha_{0} \in [\MC]}  |\lb \alpha_{0} \vert \alpha \rb|^4
\\
=
\overline
{
\cos^2[\theta_{{\cal L}_{\rmsubscr{d-}\hat{H}} ,\,  (\vert \alpha_{0} \rb\lb \alpha_{0}\vert) }]
}^{\alpha_{0}}
\eqendspace,
\qquad\mbox{}
\label{eta}
\end{multline}
%
where ${\cal L}_{\rmsubscr{d-}\hat{H}} = {\cal L}_{\rmsubscr{tl,\,d-}\hat{H}} \oplus {\cal L}_{U}$ is the space of all diagonal operators, {\it including} the identity.


Let us now turn to a particular application of our theory: namely,
to the problem of optimization of the set of conserved quantities to be used to
enhance the predictive power of thermodynamics, extended to include additional conserved quantities; both integrable systems
and the systems in between integrable and completely thermalizable will be considered. Since the first demonstration
of such enhancement, on the example  of
one-dimensional hard-core bosons \cite{rigol2007}, the question of which set of conserved quantities should be used in a general case has remained largely open. In the hard-core boson case
and in subsequent realizations \cite{cazalilla2006}, there existed a straightforward map between the
integrable system of interest and a hidden underlying system of free particles. Conventionally,
the occupation numbers of the one-body eigenstates of the latter were assumed to constrain the relaxation dynamics the most.
These were subsequently included in the Gibbs exponent producing the so-called Generalized Gibbs Ensemble (GGE)\cite{rigol2007}.

The above choice of the conserved quantities is indeed the most natural one. However, there exist several strong incentives to formalize
the choice of the conserved quantities for GGE: (a) there are numerical indications that in case of a disorder-induced {\it localization},
the one-body free-fermionic occupations do not improve the predictive power of thermodynamics at all \cite{gramsch2012_PRA};
(b) not every integrable system can be mapped to free particles; and (c$\mbox{)}$ in the intermediate systems,
between integrable and ergodic, the variations of the expectation values of observables from one eigenstate to another are,
for all practical purposes, indistinct from the ones generated by well defined integrals of motion.
It is therefore desirable to devise a thermodynamic recipe that does not rely on a priori chosen constants of motion.
Ideally, one should not even assume that constants of motion exist.
The recipe we suggest is presented below.

To begin, one identifies a linear subspace of observables whose relaxation the generalized ensemble aims to describe.
Next, one chooses a diagonal traceless observable that minimizes the Hilbert-Schmidt angle between the space of observables of interest and itself. Next, one repeats the procedure in the space of diagonal traceless observables orthogonal to the first chosen. The procedure
is repeated recursively (every time in the space orthogonal to all the integrals of motion previously chosen) till the desired
predictive power is reached.
This procedure is based on the following exact result and its corollaries, where the HS structure naturally emerges:
\begin{widetext}
\begin{equation}
%
\Var_{\GGE}[\lb\alpha\vert \hat{A} \vert\alpha\rb] \; \leq
     \Var_{\MC}[\lb\alpha\vert \hat{A} \vert\alpha\rb] \,
     \left(\sin^2[\theta_{\hat{I}_{\rmsubscr{tl}},\, \hat{\hat{P}}_{H} \hat{A}_{\rmsubscr{tl}}}]
        +
        |\cos[\theta_{\hat{I}_{\rmsubscr{tl}},\, \hat{\hat{P}}_{H} \hat{A}_{\rmsubscr{tl}}}]| \,
        \underbrace
        {
        {\cal O}\left[
      	             \frac{\Delta I}{\sqrt{\Var_{\MC}[\lb\alpha\vert \hat{I} \vert\alpha\rb]}}
      	           \right]
        }_{\ll 1}
     \right)
\label{upper_bound_000}
\eqendspace,
%
\end{equation}
\end{widetext}
where
$\hat{\hat{P}}_{H}$ is a ``super-operator'' that removes the off-diagonal (with respect to the basis of the eigenstates
of $\hat{H}$) matrix elements, $\Delta I \equiv \max_{j} (I_{j+1} - I_{j})$ is the maximal width of the microcanonical window for the additional integral of motion,
and $\{[I_{j+1} - I_{j}]\}$ is a set of intervals tiling the axis of the integral of motion $I$. (See the extensive discussion in Supplementary Material \cite{supmat}.)
Here and below, $\Var_{\GGE}[\lb\alpha\vert \hat{A} \vert\alpha\rb]$ and $\Var_{\MC}[\lb\alpha\vert \hat{A} \vert\alpha\rb]$ define the mean square error
of the (microcanonical version of) GGE and
of the microcanonical ensemble proper.

Figure~\ref{Olshanii_fig3} shows the result of an application of this procedure to another system of one-dimensional hard-core bosons, smaller than before so that finite-size effects are enhanced. The space of observables of interest was formed by all distinct matrix elements of the one-body density matrix. We analyze the momentum distribution,
considering both the integrable case (Figs.\ \ref{Olshanii_fig3}a,c) and the case where the integrability is broken by a soft-core finite-range repulsive potential (Figs.\ \ref{Olshanii_fig3}b,d). We compare three thermodynamical ensembles. The first one is the traditional microcanonical ensemble. The second is a generalized microcanonical ensemble that fixes not only the value of the energy but generally all one-body occupation numbers of the underlying  free-fermionic system (or, in the nonintegrable case, their time-averaged values). The third is an ensemble based on the first few most relevant integrals of motion chosen using the Hilbert-Schmidt optimization procedure.
Fig.\ \ref{Olshanii_fig3}c demonstrates that in the integrable case, the accuracy of the second and the third ensembles are comparable while both greatly exceed the one of the conventional thermal ensemble. This indicates that the free-fermionic integrals of motion are indeed the optimal predictors of the after-relaxation state of the system in this case. However, the free-fermionic one-body occupations are not available at all in the
perturbed case (see Fig.\ \ref{Olshanii_fig3}d). Nevertheless, the optimized generalized microcanonical ensemble remains well defined, and it fully retains its predictive power.
Remark that for the strength of perturbation chosen, the system remains relatively close to integrable, showing a high inverse participation ratio of $\eta=.29$.

Finally, in Fig.\ \ref{Olshanii_fig4} we consider a quench from the ground state in the integrable regime to a strongly perturbed regime ($\eta=.023$). Here we directly compare the infinite time average with the predictions of both the microcanonical and the optimal GGE ensembles. The predictive power of the latter
is indeed higher than of the former.

%

%



Two distinct sources of deviation from the thermal behavior are traditionally identified. The first, the mathematically elegant one,
is associated with either nontrivial symmetries or with the Bethe ansatz \cite{gaudin1983_book,sutherland2004_book}.
The second source, the empirically important, stems from the deviations
from the Eigenstate Thermalization \cite{rigol2008} in finite systems. In small systems,---such as the nano-opto-mechanical resonators
\cite{teufel2011_359,chan2011_89,verhagen2012_63,oconnel2010_697,groblacher2009_724,teufel2011_204,aspelmeyer2012_29}---the two become practically indistinct, and no obvious candidates for the relevant conserved quantities are any longer visible: here, our theory offers a unified approach, based on a ``blind'' optimization of the predictive power of thermodynamics.

%
\begin{widetext}
\mbox{}
\begin{figure}[h]
\begin{center}
\includegraphics[scale=.5]{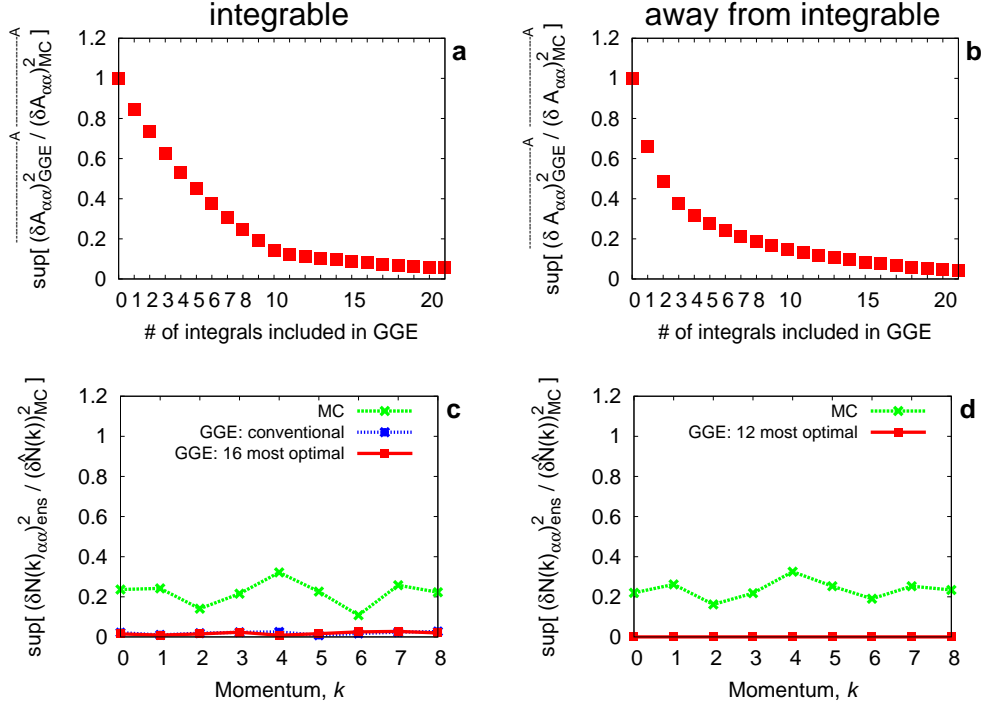}
\end{center}
\vspace{-0.7cm}
\caption {\label{Olshanii_fig3}
%
%
{\bf Predictive power of two extensions of the microcanonical ensemble}. We use the example of a system of one-dimensional hard-core bosons. For this figure: the number of particles is $N=3$, we use periodic boundary conditions, and the rest of the parameters is the same as for Fig.\ \ref{Olshanii_fig1}b.
{\bf a}, An upper bound \cite{supmat} on the mean-square error of the predictions of the Generalized Gibbs Ensemble, optimal with respect to all one-body observables (the optimal GGE), as a function of the allowed number of the additional integrals of motion involved. The error of the standard microcanonical (MC) ensemble is used as a reference. We consider the integrable case, $U=0$. The result is averaged over all one-body observables.
{\bf b}, The same as for {\bf a}, but away from integrability, with $U=.1$.
{\bf c}, An upper bound on the mean-square error of the predictions of the optimal GGE, for the
momentum distribution, in the integrable case ($U=0$). The number of the integrals of motion is fixed to the number of lattice sites, $L=16$. The results for the microcanonical ensemble, and for the conventional free-fermionic Generalized Gibbs Ensemble (GGE)\cite{rigol2007} are shown for comparison. The ensemble variance of the observables is used as a reference.
{\bf d}, The same as for {\bf c}, but with $U=.1$\,. Note that in the nonintegrable case, no free theory is available for comparison.
}
\end{figure}
\end{widetext}
%

%

\begin{figure}
\begin{center}
\includegraphics[scale=.6]{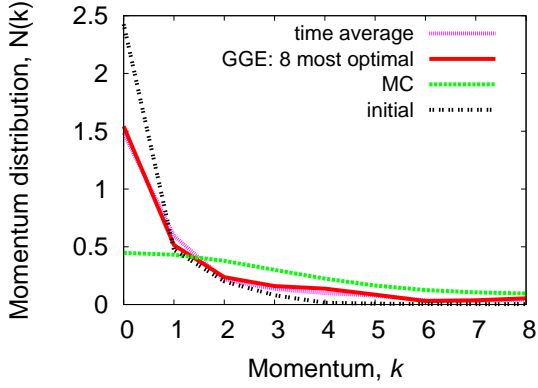}
\end{center}
\vspace{-0.7cm}
\caption{\label{Olshanii_fig4}
%
%
{\bf Momentum distribution after
a quench from the ground state of another hamiltonian}.
The initial state is the ground state of a hard-core boson hamiltonian for $N=4$ atoms on $L=16$ sites, with periodic boundary conditions. At $t=0+$, a soft-core repulsion
of strength $U=3$ is turned on. The microcanonical ensemble is represented by $N_{\MC} = 300$ lowest eigenstates.
The GGE ensemble incorporates the values of 8 most optimal integrals of motion. Each integral of motion (except the first, which was strongly
correlated with the energy) was fixed to a window around its initial state value. The half-width of each window was 10\% of the corresponding
microcanonical standard deviation.
}
\end{figure}


\section{SUPPLEMENTARY DISCUSSION}

\subsection{1. Two lemmas used to establish the upper bounds on the accuracy of the Generalized Gibbs Ensembles.}
%
%
%
%
\begin{lem}
\label{lemma_relating_the_three}
%
%
Let $\hat{q}$ be a real $n \times n$ matrix of orthogonal projection
(i.e.  let $\hat{q}$ be idempotent and symmetric: $\hat{q}\hat{q} = \hat{q}$ and $\hat{q}  = \hat{q}^{T}$, where $^{T}$ denotes transposition). Then
for any two $n$-dimensional real vectors $\vec{a}$ and $\vec{\imath}$,
\begin{eqnarray}
\|\hat{q}\vec{a}\|^2 \le
\|\vec{a}\|^2
  \left\{
	\sin^2[\theta_{\vec{\imath},\,\vec{a}}]
	+
	2 |\!\cos[\theta_{\vec{\imath},\,\vec{a}}]|\frac{\|\hat{q}\vec{\imath}\|}{\|\vec{\imath}\|}
  \right\}
\eqendspace,
\label{upper_bound_abstract}
\end{eqnarray}
where
$\theta_{\vec{v}_{1},\,\vec{v}_{2}}$ is the angle between two vectors $\vec{v}_{1}$ and $\vec{v}_{2}$, i.e.\
$\theta_{\vec{v}_{1},\,\vec{v}_{2}} \equiv \arccos[(\vec{v}_{1}\cdot\vec{v}_{2})/(\|\vec{v}_{1}\| \|\vec{v}_{2}\|)]$
and
$0 \le \theta_{\vec{v}_{1},\vec{v}_{2}} \le \pi$, and $\|\vec{v}\| \equiv \sqrt{(\vec{v}\cdot\vec{v})}$ is the norm of a vector $\vec{v}$.
%
\end{lem}%
%
\begin{proof}
Let $\hat{p} \equiv \hat{1} - \hat{q}$ be the orthogonal projection
complementary to $\hat{q}$. Similarly to $\hat{q}$, it is idempotent and symmetric.
Any vector $\vec{v}$ can be decomposed onto the sum $\vec{v} = \hat{p}\vec{v} + \hat{q}\vec{v}$, with $(\hat{p}\vec{v} \cdot \hat{q}\vec{v}) = 0$. As a consequence,
$\|\vec{v}\|^2 = \|\hat{p}\vec{v}\|^2 + \|\hat{q}\vec{v}\|^2$.
The proof goes as follows:
\begin{eqnarray}
\|\hat{q}\vec{a}\|^2
&
=
&
\|\vec{a}\|^2 - \|\hat{p}\vec{a}\|^2
\nonumber
\;\;\mbox{(using $\|\vec{a}\|^2 = \|\hat{p}\vec{a}\|^2 + \|\hat{q}\vec{a}\|^2$)}
\\
&\stackrel{\textrm{CSI}}{\le}&
\|\vec{a}\|^2 - \frac{(\vec{\imath} \cdot \hat{p}\vec{a})^2}{\|\vec{\imath}\|^2}
\nonumber
\\
&
=
&
\|\vec{a}\|^2 - \frac{(\hat{p} \vec{\imath} \cdot \vec{a})^2}{\|\vec{\imath}\|^2}
\nonumber
\;\; \mbox{(using that $\hat{p}$ is symmetric)}
\\
&
=
&
\|\vec{a}\|^2 - \frac{\left((\vec{\imath} \cdot \vec{a}) - (\hat{q} \vec{\imath} \cdot \vec{a}) \right)^2 }{\|\vec{\imath}\|^2}
\nonumber
\;\; \mbox{(using $\hat{p} = \hat{1} - \hat{q}$)}
\\
&\stackrel{\textrm{RTI}}{\le}&
\|\vec{a}\|^2 - \frac{\left( |(\vec{\imath} \cdot \vec{a})| - |(\hat{q} \vec{\imath} \cdot \vec{a})| \right)^2 }{\|\vec{\imath}\|^2}
\nonumber
\\
&=&
\|\vec{a}\|^2
-
\frac{
	   (\vec{\imath} \cdot \vec{a})^2 - 2 |(\vec{\imath} \cdot \vec{a})||(\hat{q} \vec{\imath} \cdot \vec{a})| +(\hat{q} \vec{\imath} \cdot \vec{a})^2
	 }
	 {
	   \|\vec{\imath}\|^2
	 }
\nonumber	
\\
&\le &
\|\vec{a}\|^2
-
\frac{
	   (\vec{\imath} \cdot \vec{a})^2 - 2 |(\vec{\imath} \cdot \vec{a})||(\hat{q} \vec{\imath} \cdot \vec{a})|
	 }
	 {
	   \|\vec{\imath}\|^2
	 }	
\nonumber
\\
&\stackrel{\textrm{CSI}}{\le} &
\|\vec{a}\|^2
-
\frac{
	   (\vec{\imath} \cdot \vec{a})^2 - 2 |(\vec{\imath} \cdot \vec{a})| \|\hat{q} \vec{\imath}\| \|\vec{a}\|
	 }
	 {
	   \|\vec{\imath}\|^2
	 }	
\nonumber
\\
&=&
\|\vec{a}\|^2
  \left\{
	\sin^2[\theta_{\vec{\imath},\,\vec{a}}]
	+
	2 |\!\cos[\theta_{\vec{\imath},\,\vec{a}}]|\frac{\|\hat{q}\vec{\imath}\|}{\|\vec{\imath}\|}
  \right\}
\eqendspace.
\nonumber
%
\end{eqnarray}
\end{proof}
Here, CSI stands for the Cauchy-Schwarz inequality, $|(\vec{v}_{1}\cdot\vec{v}_{2})| \le \|\vec{v}_{1}\| \|\vec{v}_{2}\| $, and
RTI denotes the reverse triangle inequality, $|x-y| \le ||x|-|y||$, $x$ and $y$ being real numbers.

%
%
%
\begin{lem}
\label{lemma_on_two_reals}
For any two real numbers $\alpha $ and $\sigma $,
\begin{multline}
%
\|\vec{a}_{\perp} + \alpha\vec{u}\|^2
\sin^2[\theta_{\vec{\imath}_{\perp} + \sigma \vec{u},\, \vec{a}_{\perp} + \alpha\vec{u}}]
\ge
\\
\|\vec{a}_{\perp}\|^2
\sin^2[\theta_{\vec{\imath}_{\perp},\, \vec{a}_{\perp}}]
\eqendspace,
\label{lemma_2}
%
\end{multline}
where $\vec{a}_{\perp}$, $\vec{\imath}_{\perp}$ are real $n$-dimensional vectors, $\vec{u}$ is a real unit $n$-dimensional vector, $\|\vec{u}\| = 1$,
and $(\vec{a}_{\perp} \cdot \vec{u}) = (\vec{\imath}_{\perp} \cdot \vec{u}) = 0$.
%
\end{lem}
%
%
\begin{proof}
Since neither the left- nor the right-hand-side of the inequality (\ref{lemma_2}) depends on the norm of $\vec{\imath}$,
we may assume, without loss of generality, that $\vec{\imath}$ is a unit norm vector: $\vec{\imath} = \vec{e}_{\vec{\imath}}$, where
$\| \vec{e}_{\vec{\imath}} \| = 1$.
It can be further decomposed as $\vec{e}_{\vec{\imath}} = \cos(\eta) (\vec{e}_{\vec{\imath}})_{\perp} + \sin(\eta) \vec{u}$, where
$(\vec{e}_{\vec{\imath}})_{\perp} \equiv \vec{e}_{\vec{\imath}} - (\vec{u} \cdot \vec{e}_{\vec{\imath}}) \vec{u}$. The proof is as follows:
\begin{eqnarray}
&&
\mbox{}\hspace{-5em}\mbox{}
\|\vec{a}_{\perp} + \alpha\vec{u}\|^2
\sin^2[\theta_{\vec{\imath},\, \vec{a}_{\perp} + \alpha\vec{u}}]
\stackrel{\vec{\imath} = \vec{e}_{\vec{\imath}}}{=}
\nonumber
\\
&\stackrel{\vec{\imath} = \vec{e}_{\vec{\imath}}}{=}&
\|\vec{a}_{\perp} + \alpha\vec{u}\|^2
\sin^2[\theta_{\vec{e}_{\vec{\imath}},\, \vec{a}_{\perp} + \alpha\vec{u}}]
\nonumber
\\
&=&
\|\vec{a}_{\perp} + \alpha\vec{u}\|^2
-
(\vec{e}_{\vec{\imath}} \cdot (\vec{a}_{\perp} + \alpha\vec{u}))
\nonumber
\\
&=&
\|\vec{a}_{\perp}\|^2
-
(\vec{e}_{\vec{\imath}} \cdot \vec{a}_{\perp})^2
+
(1-(\vec{e}_{\vec{\imath}} \cdot \vec{u})^2) \alpha^2
\nonumber \\
&& \mbox{}\hspace{10em}\mbox{}
-
2 (\vec{e}_{\vec{\imath}} \cdot \vec{a}_{\perp}) (\vec{e}_{\vec{\imath}} \cdot \vec{u}) \alpha
\nonumber
\\
&\ge&
\|\vec{a}_{\perp}\|^2
-
(\vec{e}_{\vec{\imath}} \cdot \vec{a}_{\perp})^2
\nonumber \\
&& \mbox{}\hspace{0em}\mbox{}
+
\min_{\alpha}\large[
\underbrace{(1-(\vec{e}_{\vec{\imath}} \cdot \vec{u})^2)}_{\ge 0} \alpha^2
-
2 (\vec{e}_{\vec{\imath}} \cdot \vec{a}_{\perp}) (\vec{e}_{\vec{\imath}} \cdot \vec{u}) \alpha
\large]
\nonumber
\\
&=&
\|\vec{a}_{\perp}\|^2
-
(\vec{e}_{\vec{\imath}} \cdot \vec{a}_{\perp})^2
-
\frac{(\vec{e}_{\vec{\imath}} \cdot \vec{a}_{\perp})^2 (\vec{e}_{\vec{\imath}} \cdot \vec{u})^2}{1-(\vec{e}_{\vec{\imath}} \cdot \vec{u})^2}
\nonumber
\\
&=&
\|\vec{a}_{\perp}\|^2
-
\frac{(\vec{e}_{\vec{\imath}} \cdot \vec{a}_{\perp})^2}{1-(\vec{e}_{\vec{\imath}} \cdot \vec{u})^2}
\nonumber
\\
&
=
&
\|\vec{a}_{\perp}\|^2
-
\frac{\cos^2[\eta]((\vec{e}_{\vec{\imath}})_{\perp} \cdot \vec{a}_{\perp})^2}{1-\sin^2[\eta]}
\nonumber
\\
&& \;\; \mbox{(using $\vec{e}_{\vec{\imath}} = \cos(\eta) (\vec{e}_{\vec{\imath}})_{\perp} + \sin(\eta) \vec{u}$)} \nonumber
\\
&=&
\|\vec{a}_{\perp}\|^2
-
((\vec{e}_{\vec{\imath}})_{\perp} \cdot \vec{a}_{\perp})^2
\nonumber
\\
&
=
&
\|\vec{a}_{\perp}\|^2
\sin^2[\theta_{\vec{\imath}_{\perp},\, \vec{a}_{\perp}}]
\eqendspace.
\;\; \mbox{(using $(\vec{e}_{\vec{\imath}})_{\perp} = \vec{\imath}_{\perp}$)}
\nonumber
%
\end{eqnarray}
\end{proof}
%

\subsection{2. Accuracy of a Generalized Gibbs Ensemble. A single additional integral of motion and a single observable of interest.}
Imagine that we are interested in predicting the infinite time average of the quantum-mechanical mean
of an observable $\hat{A}$, given the initial state. According to the
standard microcanonical scenario, the energy scale is divided onto narrow intervals, $[E_{i},\,E_{i+1}]$. The only information about the
initial state we are given is which energy interval, $i^{\star}$, the quantum expectation value of the energy belongs to. The microanonical
prediction for the time average of the observable in the subsequent evolution is
\begin{eqnarray}
&&
\Prediction[\Mean_{t}[A]]
=
\Mean_{\MC}[A]
\nonumber
\\
&&
\qquad\qquad\qquad\qquad\qquad\qquad
         \equiv \frac{
	                   \sum_{\alpha \in D_{i^{\star}}}\, \lb\alpha\vert \hat{A} \vert\alpha\rb
	                 }
	                 {
	                   \sum_{\alpha \in D_{i^{\star}}}\, 1
	                 }
\eqendspace,
\nonumber
\end{eqnarray}
where $D_{i}$ is the interval of the eigenstate indices $\alpha$ populated by the eigenstates whose energy belongs to the
interval $[E_{i},\,E_{i+1}]$:
\begin{eqnarray}
\alpha \in D_{i} \Leftrightarrow E_{\alpha} \in [E_{i},\,E_{i+1}]
\eqendspace.
\label{defining_D_i}
\end{eqnarray}
Now, assume that in the initial state, the system's energy is measured exactly, yielding a value $E_{\alpha^{\star}}$, but the only information we
get is, again, the interval $D_{i^{\star}}$ this energy belongs to. Assuming that within the interval $D_{i^{\star}}$ each eigenstate can
appear with equal probability, the mean square error of the microcanonical prediction for the quantum-mechanical mean of $\hat{A}$ is
\begin{eqnarray}
\Var_{\MC}[\lb\alpha\vert \hat{A} \vert\alpha\rb]
=
\frac{
       \sum_{\alpha \in D_{i^{\star}}} (\lb\alpha\vert \hat{A} \vert\alpha\rb - \Mean_{\MC}[A])^2
	 }
	 {
	   \sum_{\alpha \in D_{i^{\star}}}\, 1
	 }
\eqendspace.
\nonumber
\end{eqnarray}
The vanishing of this variance in the thermodynamic limit is the essence of the Eigenstate Thermalization Hypothesis
\cite{deutsch1991,srednicki1994,rigol2008}.

Imagine now that we are given an extra piece of information: via a second measurement in the initial state,
we are allowed to place another integral of
motion $\hat{I}$---an observable with all off-diagonal matrix elements equal to zero---to an interval $j^{\star}$, one of a set of intervals
$[I_{j},\,I_{j+1}]$. An ensemble of states with $E_{\alpha} \in [E_{i^{\star}},\,E_{i^{\star}+1}]$ and
$I_{\alpha} \in [I_{j^{\star}},\,I_{j^{\star}+1}]$,
distributed with equal probability, forms a minimal version of the Generalized Gibbs Ensemble \cite{rigol2007}. Its prediction
for the quantum mean of the observable $\hat{A}$ reads:
\begin{eqnarray}
&&
\Prediction[\Mean_{t}[A]]
=
\Mean_{\GGE\,\vert\,j^{\star}}[A]
\nonumber
\\
&&
\qquad\qquad\qquad\qquad\qquad\qquad
         \equiv \frac{
	                   \sum_{\alpha \in D_{i^{\star}} \cap S_{j^{\star}} }\, \lb\alpha\vert \hat{A} \vert\alpha\rb
	                 }
	                 {
	                   \sum_{\alpha \in D_{i^{\star}} \cap S_{j^{\star}} }\, 1
	                 }
\eqendspace,
\nonumber
\end{eqnarray}
where
\begin{eqnarray}
\alpha \in S_{j} \Leftrightarrow I_{\alpha} \in [I_{j},\,I_{j+1}]
\eqendspace.
\nonumber
\end{eqnarray}
Here and below, $I_{\alpha} \equiv \lb\alpha\vert \hat{I} \vert\alpha\rb$.

The predictions of the new ensemble are, by construction, always more accurate than the microcanonical ones.
Using Lemma~\ref{lemma_relating_the_three} from Sec.~1,
one can prove the following
%
%
%
\begin{thm}
In predicting the infinite time average of an observable $\hat{A}$,
the mean square error (MSE) of the prediction of Generalized Gibbs Ensemble based on an integral of motion $\hat{I}$,
\begin{eqnarray}
%
&& \MSE[\Prediction[\Mean_{t}[A]]] \equiv \Var_{\GGE}[\lb\alpha\vert \hat{A} \vert\alpha\rb] \nonumber \\
&& \qquad =
\frac{
	   \sum_{j} \sum_{\alpha \in D_{i^{\star}} \cap S_{j} } (\lb\alpha\vert \hat{A} \vert\alpha\rb - \Mean_{\GGE\,\vert\,j}[A])^2
	 }
	 {
	   \sum_{\alpha \in D_{i^{\star}}}\, 1
	 }
\eqendspace,
\nonumber
\end{eqnarray}
(where it is assumed that the initial states are still uniformly distributed inside $[E_{i^{\star}},\,E_{i^{\star}+1}]$)
is bounded from the above as
\begin{widetext}
%
\begin{multline}
%
\MSE[\Prediction[\Mean_{t}[A]]] \equiv \Var_{\GGE}[\lb\alpha\vert \hat{A} \vert\alpha\rb] \; \leq
\\
     \Var_{\MC}[\lb\alpha\vert \hat{A} \vert\alpha\rb] \,
     \left(\sin^2[\theta_{\hat{I}_{\rmsubscr{tl}},\, \hat{\hat{P}}_{H} \hat{A}_{\rmsubscr{tl}}}]
        +
        |\cos[\theta_{\hat{I}_{\rmsubscr{tl}},\, \hat{\hat{P}}_{H} \hat{A}_{\rmsubscr{tl}}}]| \,
        \underbrace
        {
        {\cal O}\left[
      	             \frac{\Delta I}{\sqrt{\Var_{\MC}[\lb\alpha\vert \hat{I} \vert\alpha\rb]}}
      	           \right]
        }_{\ll 1}
     \right)
\label{upper_bound}
\eqendspace,
%
\end{multline}
\end{widetext}
where
$\theta_{\hat{B}_{1},\, \hat{B}_{2}}$ is the Hilbert-Schmidt angle between the observables
$\hat{B}_{1}$ and $\hat{B}_{2}$,
a traceless version of a given observable, $\hat{B}$, is defined as $\hat{B}_{\rmsubscr{tl}} \equiv \hat{B} - \Trace[\hat{B}]/N_{\MC}$,
$\hat{\hat{P}}_{H}$ is a ``super-operator'' that removes the off-diagonal (with respect to the basis of the eigenstates
of $\hat{H}$) matrix elements,
and $\Delta I \equiv \max_{j} (I_{j+1} - I_{j})$ is the maximal width of a microcanonical window for the additional integral of motion.
%
\end{thm}
%
The estimate in Eq.~(\ref{upper_bound}) shows that the biggest increase in the predictive power of the Generalized Gibbs Ensemble is delivered by
the integrals of motion at a small Hilbert-Schmidt angle to the observable of interest.
To the dominant order in the size of the microcanonical boxes for the integral of motion $\Delta I$,
the bound does not depend on the details of the partition of the axis of the integral of motion $I$.

Remark that if the bound (\ref{upper_bound}) is used to assess the accuracy in a quench from a linear superposition of the eigenstates,
the variance of the integral of motion in the initial state must be smaller than $\Delta I$.

The dictionary used to relate Lemma~\ref{lemma_relating_the_three} to the bound (\ref{upper_bound}) is as follows:
\begin{eqnarray}
\vec{a} &\to & \{\lb\alpha\vert \hat{A} \vert\alpha\rb\ \, \vert \, \alpha \in D_{i^{\star}} \}
\nonumber
\\
\vec{\imath} &\to & \{\lb\alpha\vert \hat{I} \vert\alpha\rb\ \, \vert \, \alpha \in D_{i^{\star}} \}
\nonumber
\\
\hat{q}\vec{a} &\to & \{\lb\alpha\vert \hat{A} \vert\alpha\rb - \Mean_{\GGE\,\vert\,j(\alpha)}[A] \, \vert \, \alpha \in D_{i^{\star}} \}
\nonumber
\\
\hat{q}\vec{\imath} &\to & \{\lb\alpha\vert \hat{I} \vert\alpha\rb - \Mean_{\GGE\,\vert\,j(\alpha)}[I]  \, \vert \, \alpha \in D_{i^{\star}} \}
\nonumber
\\
\|\vec{a}\|^2 &\to & \Var_{\MC}[\lb\alpha\vert \hat{A} \vert\alpha\rb] + \Mean_{\MC}[A]^2
\nonumber
\\
\|\vec{\imath}\|^2 &\to & \Var_{\MC}[\lb\alpha\vert \hat{I} \vert\alpha\rb] + \Mean_{\MC}[I]^2
\nonumber
\\
\|\hat{q}\vec{a}\|^2 &\to & \Var_{\GGE}[\lb\alpha\vert \hat{A} \vert\alpha\rb]
\nonumber
\\
\|\hat{q}\vec{\imath}\|^2 &\to & \Var_{\GGE}[\lb\alpha\vert \hat{I} \vert\alpha\rb]
\nonumber
\eqendspace.
\end{eqnarray}
%
Lemma~\ref{lemma_on_two_reals} shows that the first term of the bound in Eq.~(\ref{upper_bound_abstract}) (which is the dominant term in our context) can be further improved by removing a component along a particular unit vector $\vec{u}$.
In our case, the role of the vector $\vec{u}$ will be played
by a Hilbert-Schmidt-normalized identity operator. The following dictionary completes the proof of Eq.~(\ref{upper_bound}):
\begin{eqnarray}
\vec{a_{\perp}} &\to & \{\lb\alpha\vert \hat{A}_{\rmsubscr{tl}} \vert\alpha\rb\ \, \vert \, \alpha \in D_{i^{\star}} \}
\nonumber
\\
\vec{\imath}_{\perp} &\to & \{\lb\alpha\vert \hat{I}_{\rmsubscr{tl}} \vert\alpha\rb\ \, \vert \, \alpha \in D_{i^{\star}} \}
\nonumber
\\
\vec{u} &\to & \{1/\sqrt{N_{\MC}} \, \vert \, \alpha \in D_{i^{\star}} \}
\nonumber
\\
\|\vec{a}_{\perp}\|^2 &\to & \Var_{\MC}[\lb\alpha\vert \hat{A} \vert\alpha\rb]
\nonumber
\\
\|\vec{\imath}_{\perp}\|^2 &\to & \Var_{\MC}[\lb\alpha\vert \hat{I} \vert\alpha\rb]
\nonumber
\\
\|\hat{q}\vec{a}_{\perp}\|^2 &\to & \Var_{\GGE}[\lb\alpha\vert \hat{A} \vert\alpha\rb]
\nonumber
\\
\|\hat{q}\vec{\imath}_{\perp}\|^2 &\to & \Var_{\GGE}[\lb\alpha\vert \hat{I} \vert\alpha\rb]
\nonumber
\eqendspace.
\end{eqnarray}
%

\subsection{3. Several additional integrals of motion and a single observable of interest.}
One can further introduce another integral of motion, $\hat{I}_2$, orthogonal to the first one:
$((\hat{I}_{2})_{\rmsubscr{tl}}\vert \hat{I}_{\rmsubscr{tl}}) = 0$. Introducing a fictitious observable
with matrix elements $\lb\alpha\vert \hat{A} \vert\alpha\rb - \Mean_{\GGE\,\vert\,j(\alpha)}[A]$, the bound
(\ref{upper_bound}) can be generalized to
\begin{eqnarray}
\MSE[\Prediction[\Mean_{t}[A]]] \equiv \Var_{\GGE}[\lb\alpha\vert \hat{A} \vert\alpha\rb]
\lesssim &&
\nonumber \\
  && \mbox{}\hspace{-21em}
  \Var_{\MC}[\lb\alpha\vert \hat{A} \vert\alpha\rb]
    \times  \nonumber \\
&&   \mbox{}\hspace{-18em}
     \left(
     1 - \cos^2[\theta_{\hat{I}_{\rmsubscr{tl}},\, \hat{\hat{P}}_{H} \hat{A}_{\rmsubscr{tl}}}] - \cos^2[\theta_{(\hat{I}_{2})_{\rmsubscr{tl}},\, \hat{\hat{P}}_{H} \hat{A}_{\rmsubscr{tl}}}]
    \right)
\eqendspace.
\nonumber
\end{eqnarray}
Furthermore, by induction, one can extend this bound to any number $N_{I}$ of (mutually orthogonal) integrals of motion:
\begin{eqnarray}
\MSE[\Prediction[\Mean_{t}[A]]] \equiv \Var_{\GGE}[\lb\alpha\vert \hat{A} \vert\alpha\rb]
\lesssim &&
\nonumber \\
&&
\mbox{}\hspace{-20em}\mbox{}
\Var_{\MC}[\lb\alpha\vert \hat{A} \vert\alpha\rb]
    \left(
     1 -
     \sum_{m=1}^{N_{I}}\cos^2[\theta_{(\hat{I}_{m})_{\rmsubscr{tl}},\, \hat{\hat{P}}_{H} \hat{A}_{\rmsubscr{tl}}}]
    \right)
\eqendspace.
\nonumber
\end{eqnarray}
For future applications, we will quote two other expressions for this bound:
\begin{eqnarray}
\MSE[\Prediction[\Mean_{t}[A]]] \equiv \Var_{\GGE}[\lb\alpha\vert \hat{A} \vert\alpha\rb]
\lesssim &&
\label{bound_100}
\\
&&
\mbox{}\hspace{-20em}\mbox{}
\Var_{\MC}[\lb\alpha\vert \hat{A} \vert\alpha\rb]
-
\Var_{\MC}[A] \sum_{m=1}^{N_{I}}\cos^2[\theta_{(\hat{I}_{m})_{\rmsubscr{tl}},\, \hat{A}_{\rmsubscr{tl}}}]
\eqendspace,
\nonumber
\end{eqnarray}
where $\Var_{\MC}[A] \equiv \Mean_{\MC}[A^2]-\Mean_{\MC}[A]^2$; and
\begin{eqnarray}
\MSE[\Prediction[\Mean_{t}[A]]] \equiv \Var_{\GGE}[\lb\alpha\vert \hat{A} \vert\alpha\rb]
\lesssim &&
\label{bound_110}
\\
&&
\mbox{}\hspace{-21em}\mbox{}
\Var_{\MC}[A] \cos^2[\theta_{{\cal L}_{\rmsubscr{tl,\,d-}\hat{H}}/\Span[\{(\hat{I}_{m})_{\rmsubscr{tl}} \;|\; m=1,\,\dots,\,N_{I} \}],\, \hat{A}_{\rmsubscr{tl}}}]
\eqendspace,
\nonumber
\end{eqnarray}
where ${\cal L}/{\cal L}_{\rmsubscr{sub}}$ is the orthogonal complement of a subspace ${\cal L}_{\rmsubscr{sub}}$ of a vector space ${\cal L}$ with respect to an inner product
$(A | B)$ (the Frobenius or Hilbert-Schmidt product $(\hat{A} | \hat{B})$ in our case).
%

\subsection{4. Several additional integrals of motion and several observables of interest.}
Now, assume that we have several observables of interest (spanning a linear space ${\cal L}_{\rmsubscr{o.i.}}$) whose long-term
behavior we want to be able to predict.  Assume further that we are given several integrals
of motion (spanning a space ${\cal L}_{I,\rmsubscr{tl}}$) we are allowed to use as thermodynamical predictors. Next, we are going to form linear spaces
${\cal L}_{\rmsubscr{o.i.,tl}}$ and ${\cal L}_{I,\rmsubscr{tl}}$ spanned by the traceless versions of the same observables. Now, let us
introduce orthonormal bases for both spaces:
\begin{eqnarray}
&&
(\hat{A}_{q})_{\rmsubscr{tl}} \quad \mbox{with} \, q=1,\,2,\,\ldots,\,N_{o.i}
\nonumber
\\
&&
((\hat{A}_{q^{}})_{\rmsubscr{tl}} | (\hat{A}_{q'})_{\rmsubscr{tl}})=\delta_{q,\,q'}
\nonumber
\\
&&
\hat({I}_{m})_{\rmsubscr{tl}} \quad \mbox{with} \, m=1,\,2,\,\ldots,\,N_{I}
\nonumber
\\
&&
((\hat{I}_{m^{}})_{\rmsubscr{tl}} | (\hat{I}_{m'})_{\rmsubscr{tl}})=\delta_{m,\,m'}
\eqendspace.
\nonumber
\end{eqnarray}
In particular, this implies that $\Var_{\MC}[A_{q}] = \Var_{\MC}[I_{m}] = 1$ for all $q$ and $m$.
Let us now form the following matrix:
\begin{eqnarray}
&&
R_{m,\,m'} = \sum_{q=1}^{N_{o.i}}
                 \cos[\theta_{(\hat{I}_{m^{}})_{\rmsubscr{tl}},\, (\hat{A}_{q})_{\rmsubscr{tl}}}]
                 \cos[\theta_{(\hat{A}_{q})_{\rmsubscr{tl}},\, (\hat{I}_{m'})_{\rmsubscr{tl}}}]
\nonumber
\\
&&
m,\,m'=1,\,2,\,\ldots,\,N_{I}
\eqendspace,
\nonumber
\end{eqnarray}
cf.\, Eqn.\ \ref{bound_100}.

It can be shown that the eigenvectors of $R_{m,\,m'}$, $\hat{\tilde{I}}_{\tilde{m}}$, when ordered in the descending order of their eigenvalues $r_{\tilde{m}}$
form a sequence such that the first $N_{opt.}$ members constitute the most optimal
Generalized Gibbs Ensemble involving $N_{opt.}$ integrals, in the sense that it minimizes, on average, the error in the prediction of the ensemble.
The requests for linear combinations of the observables of interest are supposed to be distributed according to a spherically
symmetric probability distribution, with respect to the Hilbert-Schmidt measure.

We show examples of a relevance sequence in
Figs.\ 3a and b.
In the examples considered in Fig.\ 3, the space ${\cal L}_{I,\rmsubscr{tl}}$ was represented by the whole space of traceless integrals of motion
${\cal L}_{\rmsubscr{tl,\,d-}\hat{H}}$. The number of the most optimal (relevant) integrals of motion was 16 for both Fig.\ 3c and Fig.\ 3d.

\subsection{5. The classical billiard with a soft-core scatterer.}
The classical mechanics example considered in the letter consists of a two-dimensional particle of mass $m$, moving in a rectangular billiard with periodic boundary conditions (thus topologically equivalent to a torus). As an integrability-breaking perturbation, we use a ``truncated'' $\delta $-function: a potential that consists of the first $M$ (including the zeroth one) spatial harmonics of a $\delta $-function (see Fig.\ \ref{Olshanii_figSD1}). 
\begin{figure}
%
\mbox{}\\
\mbox{}\\
\mbox{}\\
\begin{center}
\includegraphics[scale=.25]{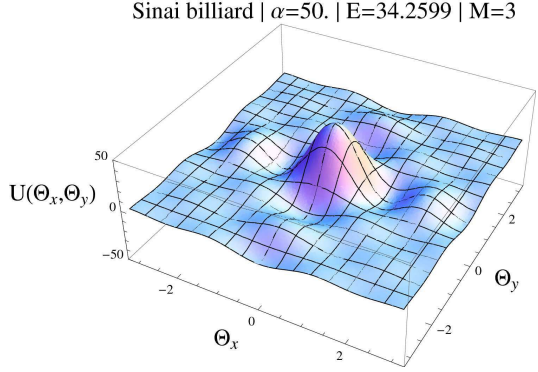}
\end{center}
\vspace{-0.3cm}
\caption{
\label{Olshanii_figSD1}
%
%
{\bf The integrability-breaking barrier for the classical system considered:}. The plot
shows the ``truncated'' $\delta $-function potential used as an integrability-breaking
perturbation acting on a particle in a two-dimensional square billiard.}
\end{figure}
%
The resulting hamiltonian is
%
\begin{multline*}
%
H(x,\,y,\,p_{x},\,p_{y}) =
  \frac{p_{x}^2}{2m} + \frac{p_{y}^2}{2m}
  \\
+
  V\,
  \frac{\sin[2\pi (M+1/2)x/L]}{(2M+1)\sin[\pi x/L] }
  \frac{\sin[2\pi (M+1/2)y/L]}{(2M+1)\sin[\pi y/L] }
\eqendspace,
\end{multline*}
where
\vspace{-1.0\baselineskip}\mbox{}
\begin{eqnarray*}
&&
0 \le x < L\,;\quad (x=L) \equiv (x=0)
\\
&&
0 \le y < L\,;\quad (y=L) \equiv (y=0)
\eqendspace.
%
\end{eqnarray*}
In what follows, we will be using a system of units where
$m = \frac{L}{2\pi} =1$.
The hamiltonian then becomes
%
\begin{multline*}
%
H(\theta_{x},\,\theta_{y},\,I_{x},\,I_{y}) =
  \frac{I_{x}^2}{2} + \frac{I_{y}^2}{2}
\\
+
  V\,
  \frac{\sin[(M+1/2) \theta_{x}]}{(2M+1)\sin[\theta_{x}/2] }
  \frac{\sin[(M+1/2) \theta_{y}]}{(2M+1)\sin[\theta_{y}/2] }
\eqendspace,
%
\end{multline*}
%
where, in this system of units, we have
$I_{\alpha} = p_{\alpha}$
and
$\theta_{\alpha} = r_{\alpha}$
($\alpha=x,\,y$); also,
$r_{x} \equiv x$; $r_{y} \equiv y$. We used $M=3$ for all data points.

For any strength of the perturbation $V$, the initial conditions were drawn
from a microcanonical ensemble bounded by the equi-energy surfaces (specific for
a given $V$) in such a way that the phase-space volume below the lower surface
was $W_{b}=7895.7$, and in between the lower and the upper was $W=1184.3$, regardless of the perturbation strength.
A quantum-mechanical analogue of such an ensemble would have the same lower and upper quantum indices of the microcanonical window
for all realizations. To give the reader an idea of the energy scale involved, we will note that for a zero barrier height, the
mean energy in the ensemble is $E_{0}=33.7$.

For the test observable $A$, we used the difference between the kinetic energies in the $x$- and $y$-directions:
$
A = \Delta E_{\textrm{kin.}} \equiv \frac{p_{x}^2}{2m} - \frac{p_{y}^2}{2m}$ (which is
$\frac{I_{x}^2}{2} - \frac{I_{y}^2}{2}$ in our system of units). We studied both the ensemble variance of the temporal mean,
\begin{widetext}
\begin{eqnarray}
%
\Var_{\MC}[\Mean_{t}[A]] \equiv
\frac
{
\int_{\MC} d\theta_{x}^{t=0} d\theta_{y}^{t=0} dI_{x}^{t=0} dI_{y}^{t=0}
   \left\{
     \left(
     \lim_{t_{\textrm{max}\to\infty}}\frac{1}{t_{\textrm{ max}}}\int_{t=0}^{t_{\textrm{max}}} \,dt\,A(t)
     \right)
     -
     \Mean_{\MC}[A]
   \right\}^2
}
{
\int_{\MC} d\theta_{x} d\theta_{y} dI_{x} dI_{y}\; 1
}
\eqendspace,
%
\end{eqnarray}
\end{widetext}
%
and the ensemble mean of the temporal variance,
\begin{widetext}
\begin{eqnarray}
&&
\Mean_{\MC}[\Var_{t}[A]] \equiv
\frac
{
\int_{\MC} d\theta_{x}^{t=0} d\theta_{y}^{t=0} dI_{x}^{t=0} dI_{y}^{t=0}      %
        \lim_{t_{\textrm{ max}\to\infty}}\frac{1}{t_{\textrm{ max}}}
        \int_{t=0}^{t_{\textrm{max}}}\,dt\,
           \left\{
           A(t)
           -
           \left( \lim_{t_{\textrm{ max}\to\infty}}\frac{1}{t_{\textrm{ max}}}\int_{t'=0}^{t_{\textrm{max}\to\infty}}\,dt'\, A(t') \right)
          \right\}^2
}
{
\int_{\MC} d\theta_{x} d\theta_{y} dI_{x} dI_{y}\; 1
}
\eqendspace,
\nonumber
\end{eqnarray}
\end{widetext}
where the ensemble average is
\begin{eqnarray}
&&
\Mean_{\MC}[A] \equiv
\frac
{
\int_{\MC} d\theta_{x} d\theta_{y} dI_{x} dI_{y}\, A
}
{
\int_{\MC} d\theta_{x} d\theta_{y} dI_{x} dI_{y}\, 1
}
\eqendspace,
\nonumber
\end{eqnarray}
$\int_{\MC} d\theta_{x} d\theta_{y} dI_{x} dI_{y} \ldots $ is an integral over the microcanonical volume describe above, and
$A(t)$ the time dependence for the observable $A$ along a trajectory that starts at $(\theta_{x}^{t=0},\,\theta_{y}^{t=0},\,I_{x}^{t=0},\,I_{y}^{t=0})$.

\subsection{6. The hamilitonian for one-dimensional hard-core bosons perturbed by soft-core interactions.}
The quantum hamiltonian used in this Letter reads
%
\begin{multline}
\hat{H}
= - J \sum_{j} (\hat{b}^{\dagger}_{j}\hat{b}^{}_{j+1} + \textrm{h.c.})
\\
         \quad + \frac{1}{2} \sum_{j_{1}}\sum_{j_{2}} V(j_{1},\ j_{2})
            \hat{b}^{\dagger}_{j_{1}}\hat{b}^{\dagger}_{j_{2}}\hat{b}^{}_{1}\hat{b}^{}_{2}
\eqendspace,
\label{H_QM}
\end{multline}
%
where the commutation relations for the hard-core boson creation and annihilation operators obey
\begin{eqnarray}
\lbrack b^{}_{j},\,b^{\dagger}_{j'} \rbrack = \lbrack b^{}_{j},\,b^{}_{j'} \rbrack
=\lbrack b^{\dagger}_{j},\,b^{\dagger}_{j'} \rbrack=0,\quad \mbox{for} \quad j'\neq j
\nonumber
\\
\lbrack b^{}_{j},\,b^{\dagger}_{j} \rbrack=1, \quad (b^{}_{j})^2=(b^{\dagger}_{j})^2=0
\eqendspace.
\nonumber
\end{eqnarray}
The soft-core interaction potential we used was a rectangular potential of a four-site range:
\begin{eqnarray}
V(j_{1},\ j_{2})=U\,
\left\{
\begin{array}{ccc}
1 & \mbox{for} & |j_{2}-j_{1}| \leq \Delta j_{\textrm{range}}
\\
0 &            & \mbox{otherwise}
\end{array}
\right.
\eqendspace,
\nonumber
\end{eqnarray}
with $\Delta j_{\textrm{range}} = 4$ (see Fig.\ \ref{Olshanii_figSD2}).
\begin{figure}
%
\begin{center}
\includegraphics[scale=.4]{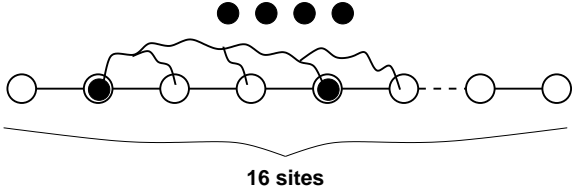}
\end{center}
\vspace{-0.7cm}
\caption{ \label{Olshanii_figSD2}
%
%
{\bf The lattice configuration for the quantum system considered}. One-dimensional hard-core
bosons on a lattice with $L=16$ sites. The number of atoms was $4$ for the system of Fig.~1b and $3$ for the system of Fig.~3. Open boundary conditions were imposed on the former system and periodic ones on the latter. The integrability-breaking perturbation was an added two-body interaction: a constant potential energy $U$ for any two atoms separated by four sites or less.}
\end{figure}

For the calculations resulting in Fig.\ 1b, the open boundary conditions were used: the first sum
in Eq.~(\ref{H_QM}) was extended to a range between $j=1$ and $j=L-1$, where $L=16$ was the length of
the lattice. Both sums in the second double sum were fixed between $j_{1,2}=1$ and $j_{1,2}=L$.

The calculations that led to Fig.\ 3 were performed using periodic boundary conditions. There, the first sum
involved in Eq.~(\ref{H_QM}) covers all sites, from $j=1$ to $j=L$. In the second sum, the index $j_{1}$ covers
the same range, from $j_{1}=1$ to $j_{1}=L$; the second index, $j_{2}$,
ranges from $j_{2}=1-\Delta j_{\textrm{range}}$
to $j_{2}=L+\Delta j_{\textrm{range}}$, so that the soft-core potential respects the periodic boundary
conditions.

Also, in both cases, a weak random on-site perturbation was used,
%
\[
W \,\sum_{j} \xi_{j} (\hat{b}^{\dagger}_{j}\hat{b}^{}_{j})
\eqendspace,
%
\]
%
where $\left\{ \xi_{j} \right\}$ is a set of independent random variables uniformly distributed
between $-1$ and $+1$, and the strength of the potential was fixed to $W=10^{-4}$. The strenghth $W$ has been
chosen in such a way that it remains weak enough not to alter any of the macroscopic properties, but strong enough
to lift all possible degeneracies and, in the periodic case, to relax the selection rules associated
with the translational invariance.

\subsection{7. Quantum observables presented in Fig.\ 1b.}
In Fig.\ 1b we analyze the properties of two integrals of motion, which are, in the case of a lattice with open boundary conditions, related to the counterparts of the fourth and sixth moments of the momentum distribution of the underlying free fermions.
When expressed through the bosonic
creation and annihilation operators,
the functionals of the fermionic momentum distribution---quadratic in the fermionic representation---become complicated many-body observables, such as
%
\begin{multline}
\hat{I}_{4}
= (1/2L) \sum_{j=1}^{L-2} (
                             (\hat{b}^{\dagger}_{j}\hat{b}^{}_{j+2} + \textrm{h.c.})
\\
                             -
                              2(
                                   \hat{b}^{\dagger}_{j}\hat{b}^{}_{j+1}\hat{b}^{\dagger}_{j+1}\hat{b}^{}_{j+2}
                                   +
                                   \hat{b}^{\dagger}_{j+2}\hat{b}^{}_{j+1}\hat{b}^{\dagger}_{j+1}\hat{b}^{}_{j}
                               )
\\
                             -
                             (\hat{b}^{\dagger}_{1}\hat{b}^{}_{1} + \hat{b}^{\dagger}_{L}\hat{b}^{}_{L})
\eqendspace.
\nonumber
\end{multline}
%
The second term inside the sum (the four-body one) is absent in the fermionic representation. The last term, present in both representations,
is responsible for the finite size effects originating, in turn, from the open boundary conditions.

We also studied a ``generic'' integral of motion: $\hat{I}_{\textrm{random}} = \sum_{\alpha_{0}} \xi_{\alpha_{0}} \vert \alpha_{0} \rb\lb \alpha_{0} \vert $,
where the states $\vert \alpha_{0} \rb $ are the eigenstates of the hamiltonian in Eq.~(\ref{H_QM}) in the absence of the integrability-breaking perturbation,
and the independent random coefficients $\xi_{\alpha_{0}}$ were uniformly distributed between $-1$ and $+1$.

\subsection{8. Quantum observables presented in Fig.\ 3.}
In Fig.\ 3, we are comparing three thermodynamic ensembles: the microcanonical ensemble (which is based on the energy alone),
\begin{eqnarray}
{\cal L}_{\rmsubscr{o.i.}} = \Span[\{\hat{H}\}]
\eqendspace,
\nonumber
\end{eqnarray}
the conventional generalized Gibbs ensemble
\cite{rigol2007} (which is based on the occupation numbers of all free-fermionic one-body orbitals, or, in the case of periodic boundary conditions,
all moments of the fermionic momentum distribution),
%
\begin{multline*}
{\cal L}_{\rmsubscr{o.i.}} = \Span\left[\left\{
	\frac{1}{2L} \sum_{j=1}^{L} (
                             (\hat{a}^{\dagger}_{j}\hat{a}^{}_{j+m} + \textrm{h.c.})
                          )
                          ,
\right.\right.
\\
\left.\left.
	\frac{-i}{2L} \sum_{j=1}^{L} (
                             (\hat{a}^{\dagger}_{j}\hat{a}^{}_{j+m} - \textrm{h.c.})
                          )
	                      \;\Big|\; m=1,\,\dots,\, L/2
	                      \right\}\right]
\eqendspace,
\nonumber
\end{multline*}
%
and the optimized generalized Gibbs ensemble (which is designed to maximize the quality of prediction
for one-body observables),
%
\begin{multline*}
{\cal L}_{\rmsubscr{o.i.}} = \Span\left[\left\{ \frac{1}{2} (\hat{b}^{\dagger}_{j'}\hat{b}^{}_{j} + \textrm{h.c.})
                          ,
\right.\right.
\\
\left.\left.
                          \frac{-i}{2} (\hat{b}^{\dagger}_{j'}\hat{b}^{}_{j} - \textrm{h.c.})
	                      \;\Big|\; j=1,\,\dots,\, L; j'=j,\,\dots,\, L
	                      \right\}\right]
\eqendspace.
\nonumber
\end{multline*}
%
Only the first few most optimal integrals of motion generated by the latter space were used.
Here $\hat{a}^{}_{j}$ are the free-fermionic annihilation operators, related to their
bosonic counterparts via a Jordan-Wigner map,
$\hat{a}^{}_{j} = \left(\prod_{j'=1}^{j-1} e^{i \pi \hat{b}^{\dagger}_{j'}\hat{b}^{}_{j'}}\right)\hat{b}^{}_{j}$. Here and below,
we are assuming that the number of sites $L$ is even.

To compare the predictive powers of these ensembles, we compute their predictions for the momentum
distribution of the bosons,
\begin{eqnarray}
&&
\hat{A}_{m} =
	\frac{1}{2L} \sum_{j=1}^{L} (
                             \hat{b}^{\dagger}_{j}\hat{b}^{}_{j+m} + \textrm{h.c.}
                          )
\nonumber
\\
&&
m=1,\,\dots,\, L/2
\eqendspace.
\nonumber
\end{eqnarray}
%


\begin{addendum}
 \item[Acknowledgements] $\!\!\!\!$ The author thanks Marcos Rigol, David Weiss, Vanja Dunjko, Alessandro Silva, and Bala Sundaram for comments. This work was supported by the
     US National Science Foundation Grant No.\ PHY-1019197,
     the Office of Naval Research Grants No.\ N00014-09-1-0502 and N00014-12-1-0400, and a grant from the {\it Institut
     Francilien de Recherche sur les Atomes Froids} (IFRAF). Support from the
     {\it Laboratoire de Physique des Lasers} of Paris 13 University is also appreciated.
 \item[Author contributions] $\!\!\!\!$ All work was done by M.O.
 \item[Additional information] $\!\!\!\!$ Supplementary information is available in the online version of the paper. Reprints and permissions information is available online at www.nature.com/reprints. Correspondence and requests for materials should be addressed to M.O.
\item[Competing financial interests] $\!\!\!\!$ The author declares no competing financial interests.

\end{addendum}
\end{document}